\documentclass[aps,twocolumn,notitlepage,pra,superscriptaddress,footinbib]{revtex4-1}

\usepackage[usenames,dvipsnames]{color}

\usepackage{gensymb} 
\usepackage{physics} 
\usepackage{amsmath} 
\usepackage{amssymb} 
\usepackage{bm} 
\usepackage{bbm} 
\usepackage{braket} 
\usepackage[version=4]{mhchem} 
\usepackage{latexsym} 
\usepackage{tensor}
\usepackage{tcolorbox}
\usepackage{mathtools}
\usepackage{mathrsfs}

\usepackage{amsthm,thmtools}
\declaretheorem[
]{theorem}

\usepackage{xr}

\newcommand{\bs}[1]{\pmb{#1}} 

\usepackage{graphicx} 
\graphicspath{{./figures/}} 
\usepackage{stackengine} 
\usepackage[caption=false]{subfig} 
\usepackage[inline]{enumitem} 
\setlist[enumerate,1]{label={(\roman*)}} 
\setlist{nolistsep} 
\usepackage{comment} 

\usepackage{multirow}
\usepackage{hhline}

\usepackage{CJKutf8}

\usepackage{hyperref}
\usepackage{natbib}
\hypersetup{
 breaklinks=true,
 colorlinks=true,
 linkcolor=red,
 filecolor=magenta,
 urlcolor=cyan,
 citecolor=blue,
}
\usepackage[capitalize]{cleveref}

\begin{document}
\begin{CJK}{UTF8}{gbsn}
\renewcommand{\labelenumii}{\theenumii}
\renewcommand{\theenumii}{\theenumi.\arabic{enumii}.}

\newcommand{\QuICS}{Joint Center for Quantum Information and Computer Science, National Institute of Standards and Technology and
 University of Maryland, College Park, Maryland 20742, USA}
\newcommand{\JQI}{Joint Quantum Institute, National Institute of Standards and Technology and
 University of Maryland, College Park, Maryland 20742, USA}
\newcommand{\USNA}{Volgenau Department of Physics, United States Naval Academy, Annapolis, MD 21402, USA}

\newcommand{\thetitle}{Error-corrected function estimation advantage in multiparameter Hamiltonians}

\title{\thetitle}

\author{Erfan~Abbasgholinejad}
\affiliation{\QuICS}
\affiliation{\JQI}
\author{Lorcán~O.~Conlon}
\affiliation{\QuICS}
\affiliation{\JQI}
\author{Sean~R.~Muleady}
\affiliation{\QuICS}
\affiliation{\JQI}
\author{Jacob~Bringewatt}
\affiliation{\USNA}
\author{Ali~Fahimniya}
\affiliation{\QuICS}
\affiliation{\JQI}
\author{Yu-Xin Wang (王语馨)}
\affiliation{\QuICS}
\author{Alexey V. Gorshkov}
\affiliation{\QuICS}
\affiliation{\JQI}

\date{\today}

\begin{abstract}
We establish the ultimate precision limits for estimating a function of multiple Hamiltonian parameters in the presence of Markovian noise. By reducing multiparameter function estimation to an optimization over effective single-parameter embeddings, we derive tight bounds on the quantum Fisher information. We identify a necessary and sufficient functional Hamiltonian-not-in-Lindblad-span condition for Heisenberg-limited scaling in time. When this condition holds, we construct a code that simultaneously removes the noise and nuisance Hamiltonian parameters while preserving the target signal. When it fails, we derive the optimal standard quantum-limit coefficient and show that it is asymptotically attainable using approximate quantum error correction. Finally, we demonstrate that direct function estimation can substantially outperform estimating all parameters individually and subsequently evaluating the function. This advantage can scale with both the evolution time and the number of parameters.
\end{abstract}

 \maketitle
\end{CJK}

\section{Introduction}

Quantum sensing exploits the controlled evolution of quantum systems to infer physical quantities with high precision \cite{Giovannetti2011,Degen2017}. A central objective of quantum metrology is to determine how the estimation error scales with the total sensing time $t$. For the mean-squared error, independent repetitions yield the standard-quantum-limit (SQL) scaling $t^{-1}$, whereas coherent phase accumulation can attain the Heisenberg-limit (HL) scaling $t^{-2}$ \cite{Giovannetti2011,Demkowicz2012}. Determining which scaling is achievable requires optimizing the available sensing strategy, including the probe state, the measurement, and, for dynamical models, the controls applied during signal acquisition.

For estimating a single parameter encoded in a Hamiltonian, the ultimate achievable precision is set by the spectral range of the generator coupled to the parameter \cite{Boixo2007,PangBrun2014}, and control can be essential when this generator is time-dependent \cite{PangJordan2017}. In the presence of Markovian noise, the asymptotic scaling of the precision with time can reduce to SQL from the HL in the noiseless case \cite{Escher2011,Demkowicz2012}. With sufficiently fast control and noiseless ancillas, however, quantum error correction (QEC) can restore HL scaling when the signal generator is distinguishable from the noise, as formalized by the Hamiltonian-not-in-Lindblad-span (HNLS) condition \cite{Arrad2014,Dur2014,Kessler2014,Demkowicz2017,Zhou2018,ZhouJiang2021Asymptotic, layden_spatial_2018,WanLasenby2022}.

The problem becomes substantially more difficult when attempting to simultaneously estimate multiple parameters, which are encoded through distinct, possibly noncommuting, Hamiltonian generators. Optimal strategies for different parameters can be incompatible, and although sensitivity bounds such as the Holevo and Nagaoka--Hayashi bounds apply to fixed quantum states \cite{Ragy2016,Holevo2011,Albarelli2019,Conlon2021,Sidhu2021, zhou2026randomized}, no general attainable bound is known after optimizing over probe states, ancillas, measurements, and controls \cite{Li2026OptimalStrategies}. 

An intermediate task between the full multiparameter and single-parameter problems is estimating a \emph{single} function of multiple parameters. Many sensing applications require only a global feature of the parameters encoded in the Hamiltonian such as a weighted average, spatial gradient, or image feature \cite{Eldredge2018, Komar2014,Taylor2008, Kaubruegger2023, Bringewatt2024, Abbasgholinejad2025, Proctor2018,Ehrenberg2023}; this task-specific viewpoint~\cite{Eldredge2018} has evolved into the recent concept of quantum computational sensing \cite{Khan2025, Sarovar2023, Allen2025}. Rather than first estimating all the parameters in the signal Hamiltonian and then evaluating the desired output in classical post-processing, a quantum computational sensor coherently processes the function of interest and avoids learning nuisance directions in the parameter space. Thus, this approach can reach a target precision using less sensing time \cite{Bringewatt2024, Zhuang2019, Banchi2020, SinananSingh2024, Allen2025, Khan2026, Sen2026}. Despite its multiparameter nature, tight bounds have been derived for sensing arbitrary functions of parameters for general, possibly noncommuting, Hamiltonians \cite{Abbasgholinejad2026}. However, the robustness of function estimation advantage to noise remains poorly understood \cite{qiao2026distributed}. 

Here, we derive precision bounds for estimating a function of multiple Hamiltonian parameters in the presence of Markovian noise, extending the theory of single-parameter noisy quantum metrology to task-specific multiparameter problems [\cref{fig:qcs}~(a)]. We give conditions that determine whether the optimal mean-squared error exhibits SQL or HL scaling in time. We establish protocols using QEC that attain the corresponding tight bounds, and characterize the roles of ancillas and intermediate controls. We then compare direct function estimation with the indirect strategy of learning the Hamiltonian parameters and subsequently evaluating the function classically. We identify regimes in which direct estimation changes the attainable time scaling from SQL to HL, as well as regimes with scaling advantage in the number of parameters. Such parameter-number advantages can persist in the presence of noise, including settings in which HL scaling in time is impossible even for the direct function estimation task. Our results therefore delineate when function estimation offers a genuine advantage over Hamiltonian learning strategies in the asymptotic time limit \cite{Huang2023, Dutkiewicz2024Control, Brahmachari2026Static}.

\begin{figure}[htpb]
	\centering
	\includegraphics[width=\linewidth]{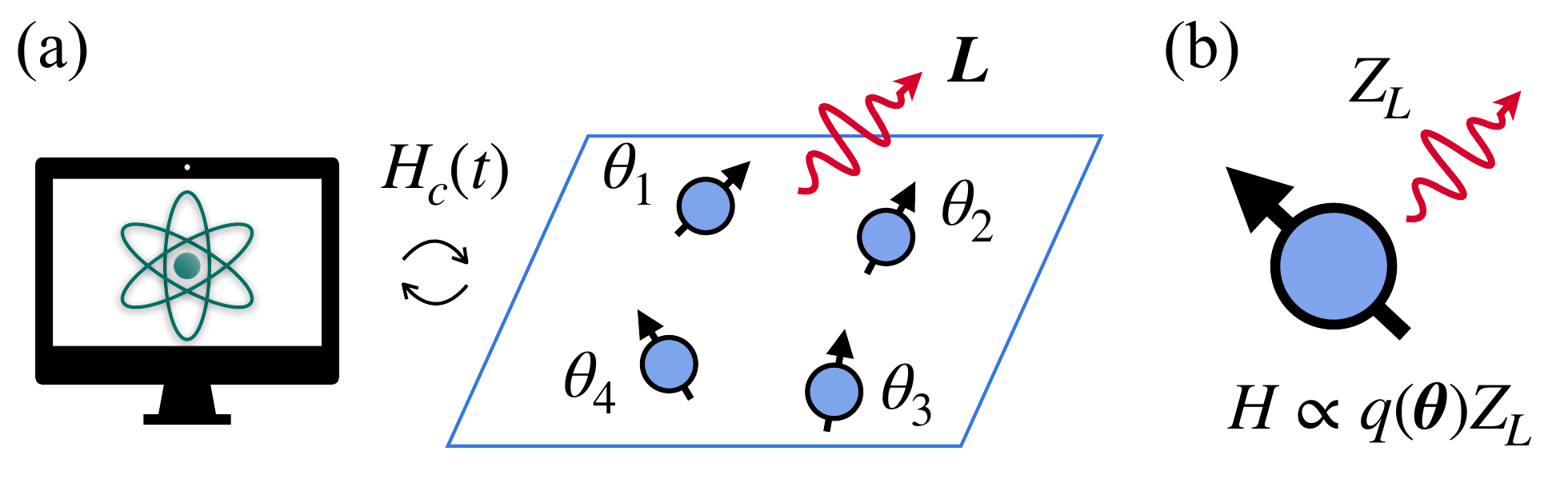}
	\caption{(a) Schematic of multiparameter function estimation with Markovian noise. A multi-qubit sensor with unknown Hamiltonian parameters $\bs\theta$ and jump operators $\bs L$ is coupled through a parameter-independent control Hamiltonian $H_c(t)$ to a noiseless ancillary quantum processor. The objective is to estimate a function $q(\bs\theta)$ without reconstructing the full parameter vector. (b) Optimized control and error correction reduce the dynamics to an effective single-qubit, single-parameter sensing problem with logical Hamiltonian $\propto q(\bs\theta)Z_L$ and, in general, logical dephasing noise. The functional HNLS condition (see \cref{thm:HNLS}) determines the absence or presence of this logical noise or, equivalently, whether sensing $q$ exhibits HL or SQL time scaling.
}
    \label{fig:qcs}
\end{figure}

\section{Model and lower bounds}

We consider a system evolving according to the Lindblad master equation
\begin{equation}\label{eq:master}
\frac{d\rho}{ds}
=-i[H(s),\rho]+\sum_a\mathcal{D}_{L_a}(\rho),
\end{equation}
where $s\in[0,t]$, $\bs L = (L_1,\cdots,L_r)$ are the jump operators describing various Markovian noise processes, and $\mathcal{D}_{L_a}(\cdot)=L_a(\cdot)L_a^\dagger-\{L_a^\dagger L_a/2, \cdot\}$. The Hamiltonian is
\begin{equation}
H(s)=H_0+H_c(s),
\qquad
H_0=\sum_{j=1}^m\theta_j g_j,
\end{equation}
where $\bs\theta\in\mathbb{R}^m$ is the vector of unknown parameters, $\bs g=(g_1,\ldots,g_m)$ is the vector of known Hermitian generators, and $H_c(s)$ is a parameter-independent control Hamiltonian acting on the system and arbitrary noiseless ancillas. Our goal is to estimate $q=\bs\alpha\cdot\bs\theta$ for a given $\bs\alpha\in\mathbb{R}^m$ and total sensing time $t$, with the variance of a locally unbiased estimator $q_\text{est}$ as the figure of merit.

We obtain lower bounds by reducing the multiparameter problem to a family of single-parameter problems \cite{Proctor2018,Eldredge2018,Abbasgholinejad2026}. Consider an invertible linear transformation from $\bs\theta$ to a new parameter vector $\bs q$, with $q_j=\bs\alpha^{(j)}\cdot\bs\theta$, $q_1=q$, and $\bs\alpha^{(1)}=\bs\alpha$. Let $\{\bs\beta^{(j)}\}_{j=1}^m$ denote the dual basis, satisfying $\bs\beta^{(j)}\cdot\bs\alpha^{(i)}=\delta_{ij}$, and define $\bs\beta\equiv\bs\beta^{(1)}$. The signal Hamiltonian can then be written as
\begin{equation} \label{eq:H_q}
H_0
=qG(\bs\beta)
+\sum_{j\geq 2}^m q_j
G(\bs\beta^{(j)}),
\end{equation}
where we define the generator map $G(\bs \beta') = \bs \beta'\cdot \bs g$. Thus, $q$ couples to an effective generator $ G(\bs\beta)$, with $\bs\alpha\cdot\bs\beta=1$.

For any such reparameterization, suppose that the nuisance parameters $\{q_j\}_{j=2}^m$ are known. Their contributions to \cref{eq:H_q} are then known and may be absorbed into the control Hamiltonian, leaving a single-parameter problem in which $q$ couples to $G(\bs\beta)$. Providing this side information cannot reduce the attainable sensitivity. Therefore, for every $\bs\beta$, the single-parameter bound on the quantum Fisher information of $q$ coupled to $G(\bs\beta)$ also bounds the information obtainable about $q$ in the original multiparameter problem. Thus, for the quantum Fisher information (QFI) of $q$ for any protocol, we obtain
\begin{equation}\label{eq:F_q}
    F_q \leq F_{q,\text{opt}}:= \min_{\bs \beta |\bs\alpha\cdot\bs \beta =1} 
    \begin{pmatrix}
        \text{single-parameter bound}\\
         \text{with generator } G(\bs\beta)
    \end{pmatrix},
\end{equation}
or equivalently, a lower bound on $\mathrm{var}(q_{\text{est}})$, through the quantum Cram\'er-Rao bound \cite{BraunsteinCaves1994, Helstrom1976}
\begin{equation}
    \mathrm{var}(q_{\text{est}}) \geq \frac{1}{F_{q,\text{opt}}}.
\end{equation} 
We emphasize that it is not immediate whether the bound in \cref{eq:F_q} is tight or not. In the absence of noise, however, tightness has recently been established for general Hamiltonian generators \cite{Abbasgholinejad2026}. In this noiseless setting, the single-parameter bound of Ref.~\cite{Boixo2007} and \cref{eq:F_q} gives
\begin{equation}\label{eq:F_q_ideal}
F_q
\leq
\min_{\bs\beta,\bs\alpha\cdot\bs\beta=1}
t^2\| G(\bs\beta)\|_s^2,
\end{equation}
where $\|\cdot\|_s$ denotes the spectral range. Ref.~\cite{Abbasgholinejad2026} constructs a protocol attaining this bound, thereby proving its tightness in the noiseless setting. We will use different choices of single-parameter bounds in \cref{eq:F_q} to treat different noise settings \cite{Demkowicz2017,WanLasenby2022, Zhou2018}.

An important feature of the bound in \cref{eq:F_q} is that the constraint $\bs\alpha\cdot\bs\beta=1$ is linear. Consequently, whenever the relevant single-parameter bound admits a semidefinite-program (SDP) representation, the additional optimization over $\bs\beta$ can potentially be incorporated into the same SDP. This structure will allow us to evaluate the noisy multiparameter bounds and, through SDP duality, construct protocols that attain them. 

\section{Error-corrected function estimation in the absence of noise}

Ref.~\cite{Abbasgholinejad2026} introduced a protocol that saturates the noiseless bound in \cref{eq:F_q_ideal} using fast coherent control during the evolution, without coupling to ancillas. Here, we present an alternative optimal protocol based on QEC for metrology \cite{Zhou2018, Demkowicz2017}. Besides providing a different perspective on the noiseless problem, this construction naturally extends to the Markovian setting considered later. In general, our protocol requires noiseless ancillas but replaces fast time-dependent control with a static control Hamiltonian.

The central idea is to treat the nuisance generators $G(\bs\beta^{(j)})$ in \cref{eq:H_q} as correctable errors. Specifically, we seek a code space $C$, with projector $\Pi_C$, satisfying the Knill--Laflamme conditions \cite{KnillLaflamme1997}
\begin{subequations}
\begin{equation}\label{eq:KL}
\Pi_C G(\bs\beta^{(j)})\Pi_C
\propto \Pi_C,
\qquad j\geq2,
\end{equation}
while preserving the signal generator,
\begin{equation}\label{eq:KL_signal}
\Pi_C G(\bs\beta)\Pi_C
\not\propto \Pi_C.
\end{equation}
\end{subequations}
\cref{eq:KL} ensures that the nuisance generators act trivially within the code space, while \cref{eq:KL_signal} guarantees that the logical dynamics remain sensitive to $q$.

The optimal code can be constructed from the dual of the minimization problem in \cref{eq:F_q_ideal}. As detailed in \cref{app:dual_HL_bound}, the dual problem is
\begin{align}\label{eq:F_opt_dual}
    F_{q,\text{opt}} =  \max_{\xi, \rho_0, \rho_1}\quad & t^2\xi^2\\
    \text{s.t.}\quad &
    \mathrm{Tr}\!\left[(\rho_0-\rho_1) g_j\right] = \xi \alpha_j,\quad \forall j,
    \nonumber\\
    & \rho_0,\rho_1\succeq 0,\qquad \mathrm{Tr}( \rho_0)=\mathrm{Tr}( \rho_1)=1. \nonumber
\end{align}
Let $\xi^*$, $\rho_0^*$, and $\rho_1^*$ denote an optimal solution. We choose the logical codewords $\ket{0_L}$ and $\ket{1_L}$ to be purifications of $\rho_0^*$ and $\rho_1^*$ using some ancilla system. These purifications can be chosen to have support on orthogonal ancillary subspaces. Then, it is easy to check that $\Pi_C = \ketbra{0_L}{0_L} + \ketbra{1_L}{1_L}$ satisfies the conditions in \cref{eq:KL,eq:KL_signal} by using the dual constraints in \cref{eq:F_opt_dual} (also detailed in \cref{app:dual_HL_bound}). Therefore, the effective Hamiltonian after projecting into the code space is 
\begin{equation}\label{eq:H_eff}
    H_\mathrm{eff}= \Pi_CH_0\Pi_C = \mu\Pi_C+\frac{\xi^* q}{2} Z_L,
\end{equation}
where $\mu\in\mathbb{R}$, and $Z_L = \ketbra{0_L}{0_L}-\ketbra{1_L}{1_L}$ is the logical Pauli-$Z$ operator. The original multiparameter problem is thus reduced to sensing $q$ with a single logical qubit [\cref{fig:qcs}~(b)]. Preparing the logical state $\ket{+_L} = (\ket{0_L}+\ket{1_L})/\sqrt{2}$ and repeatedly projecting the evolved state back into $C$ produces a logical phase generated by \cref{eq:H_eff}. In the limit of sufficiently frequent projections, the resulting quantum Fisher information is $t^2\xi^{*2}$, saturating the optimum in \cref{eq:F_opt_dual}.

Because the errors are coherent nuisance terms in the original Hamiltonian, rather than Markovian noise, the repeated projections can be replaced by a static energy penalty. In particular, adding
\begin{equation}
H_c=B\Pi_C,
\end{equation}
creates an energy gap of magnitude $B$ between the code space and its orthogonal complement. For sufficiently large $B$, leakage out of $C$ is suppressed and the dynamics are confined to the effective Hamiltonian in \cref{eq:H_eff}. In \cref{app:energy_penalty}, we show that choosing $B\gg t\|H_0\|^2$ ensures that the bias and error in the variance of $q_\text{est}$ are negligible.

We emphasize that compared with the protocol of Ref.~\cite{Abbasgholinejad2026}, this protocol eliminates the need for fast control during the sensing evolution, requiring only a static energy penalty. The tradeoff is that noiseless ancillas are generally needed to realize the logical codewords.

\section{HL function estimation in the presence of noise}

We now consider function estimation in the presence of Markovian noise. To apply the single-parameter bounds of Refs.~\cite{Zhou2018,Demkowicz2017} within the reduction of \cref{eq:F_q}, we first define the Lindblad span associated with the jump operators in \cref{eq:master}:
\begin{equation}
\mathcal S_L
=
\mathrm{span}_{\mathbb{R}}
\big\{
 I,
 L_a^{\rm H},
i L_a^{\rm AH},
( L_a^\dagger  L_b)^{\rm H},
i( L_a^\dagger  L_b)^{\rm AH}
\big\}_{a,b},
\label{eq:lindblad_span}
\end{equation}
where $ O^{\rm H}=( O+ O^\dagger)/2$ and $ O^{\rm AH}=( O- O^\dagger)/2$. The Lindblad span $\mathcal S_L$ is therefore a real subspace of the Hermitian operators.

Let $\bs g_\perp$ denote the projection of the generator vector $\bs g$ onto the orthogonal complement of $\mathcal S_L$. Thus, $g_{j,\perp}$ is the component of $g_j$ outside the Lindblad span. We define the corresponding generator map as $G_\perp(\bs\beta')=\bs\beta'\cdot\bs g_\perp$. The following theorem characterizes when the desired function retains HL scaling:

\begin{theorem}[Functional HNLS]\label{thm:HNLS}
The function $q=\bs\alpha\cdot \bs \theta$ admits HL scaling under \cref{eq:master} if and only if
\begin{equation}
\mathrm{ker}(G_\perp) \subseteq \mathrm{ker}(\bs\alpha^T),
\end{equation}
or equivalently,
\begin{equation}
\quad G_\perp(\bs v)=0 \implies \bs\alpha\cdot\bs v =0.
\end{equation}
When this condition holds, the optimal sensitivity for $q$ is attainable using an exact QEC that corrects the Markovian noise and every nuisance Hamiltonian direction.
\end{theorem}
\begin{proof}
To derive a bound and prove the theorem, we apply the single-parameter results of Refs.~\cite{Zhou2018,Demkowicz2017,ZhouJiang2021Asymptotic} to each effective generator $G(\bs\beta)$ appearing in \cref{eq:F_q}. This gives
\begin{equation}\label{eq:F_q_HL}
\lim_{t\to\infty} \frac{F_q}{t^2} \leq \min_{\bs \beta |\bs\alpha\cdot\bs \beta =1}  \min_{S\in\mathcal{S}_L} \|G(\bs\beta)-S\|_s^2.
\end{equation}
The right-hand side vanishes if and only if there exists an admissible $\bs\beta$ such that $G_\perp(\bs\beta)=0$. If the functional HNLS condition is violated, there exists a vector $\bs v\in\mathrm{ker}(G_\perp)$ with $\bs\alpha\cdot\bs v\neq0$. Rescaling $\bs v$ so that $\bs\alpha\cdot\bs v=1$ gives an admissible choice of $\bs\beta$ for which the bound vanishes, ruling out HL scaling. Conversely, if every vector in $\mathrm{ker}(G_\perp)$ is orthogonal to $\bs\alpha$, no admissible $\bs\beta$ lies in the kernel of $G_\perp$, and the coefficient on the right-hand side of \cref{eq:F_q_HL} is nonzero.

The optimization in \cref{eq:F_q_HL} is again an SDP. Its dual is (see \cref{app:dual_HL_bound})
\begin{align}\label{eq:F_opt_dual_markovian}
    F_{q,\text{opt}} =  \max_{\zeta,\rho_0,\rho_1}\quad
    & t^2 \zeta^2 \\
    \mathrm{s.t.}\quad&
    \mathrm{Tr}\left[
        (\rho_0-\rho_1) g_{j}
    \right]
    =
    \zeta\alpha_j,
    \qquad \forall j, \nonumber\\
    &\mathrm{Tr}\left[
        (\rho_0-\rho_1) S
    \right]
    =
    0,
    \qquad \forall S \in \mathcal{S}_L,\nonumber\\
    &\rho_0,\rho_1\succeq0,
    \qquad
    \mathrm{Tr}(\rho_0)=\mathrm{Tr}(\rho_1)=1. \nonumber
\end{align}
Let $\zeta^*$, $\rho_0^*$, and $\rho_1^*$ denote an optimal solution. As in the noiseless case, the logical codewords $\ket{0_L}$ and $\ket{1_L}$ can be chosen as purifications of $\rho_0^*$ and $\rho_1^*$ with support on an ancillary system. The first set of dual constraints ensures that every nuisance Hamiltonian direction acts trivially within the code space, while the desired function generates a logical Hamiltonian of the form in \cref{eq:H_eff}, with $\xi^*$ replaced by $\zeta^*$.

The additional constraint involving $\mathcal S_L$ ensures that the code also satisfies
\begin{subequations}
\begin{align}
\Pi_C L_a \Pi_C &\propto \Pi_C, \quad \forall a,\\
\Pi_C L_a^\dagger L_b \Pi_C &\propto \Pi_C, \quad \forall a,b
\end{align}
\end{subequations}
and therefore corrects the Markovian noise when QEC is applied during the evolution. The resulting logical dynamics depend only on $q$ and attain the bound in \cref{eq:F_q_HL}. This establishes both the necessity and the sufficiency of the functional HNLS condition. Moreover, $\zeta^*\leq\xi^*$ because the feasible set of \cref{eq:F_opt_dual_markovian} is contained within that of the noiseless dual problem in \cref{eq:F_opt_dual}.
\end{proof}

For estimating a single parameter $\theta_j$ with HL scaling, \cref{thm:HNLS} requires its perpendicular generator $g_{j,\perp}$ to be (i) nonzero and (ii) linearly independent of the remaining perpendicular generators. The first requirement is the usual single-parameter HNLS condition \cite{Zhou2018}, whereas the second arises due to the other nuisance parameters. Thus, the usual HNLS condition may hold while nuisance parameters still prevent HL scaling for a single parameter $\theta_j$. 

Moreover, the functional HNLS condition is weaker than the corresponding requirement for estimating the full parameter vector. Ref.~\cite{Gorecki2020OptimalProbes} showed that all Hamiltonian parameters can exhibit HL scaling only when $\mathrm{ker}(G_\perp)=\{0\}$, or equivalently when the perpendicular generators are linearly independent. For a single function, linear dependence is allowed provided that every invisible parameter direction in $\mathrm{ker}(G_\perp)$ is orthogonal to $\bs\alpha$. Consequently, the full vector $\bs\theta$ may fail to admit HL scaling even though a desired function $q=\bs\alpha\cdot\bs\theta$ remains Heisenberg limited.

The same geometric interpretation extends to several functions. If the desired functions are specified by vectors $\bs\alpha^{(k)}$, then, for each of them to exhibit HL scaling, every invisible direction in $\mathrm{ker}(G_\perp)$ must be orthogonal to each $\bs\alpha^{(k)}$. In particular, if $G_\perp$ has rank smaller than $m$, the full parameter vector cannot be estimated with HL scaling, whereas a space of linear functions with dimension $\mathrm{rank}(G_\perp)$ can still satisfy the functional HNLS condition.

\section{Approximate error correction for SQL function estimation}

We now consider the complementary regime in which the functional HNLS condition fails. In this case, $F_q$ can grow at most linearly with the total sensing time $t$. Although exact QEC cannot restore HL scaling, approximate quantum error correction (AQEC) can still optimize the SQL prefactor. The corresponding single-parameter problem was solved in Ref.~\cite{ZhouJiang2020AQEC}; here, we extend that construction to function estimation.

Applying the single-parameter bound of Ref.~\cite{ZhouJiang2020AQEC} to \cref{eq:F_q} gives
\begin{equation}\label{eq:Fq_opt_aqec}
\lim_{t\to\infty} \frac{F_{q}}{t}
\leq \frac{1}{\Gamma^*} := 4
\min_{\bs\beta|\bs \alpha \cdot \bs \beta =1} \min_{h_0,\mathbf h,\mathfrak h|\mathcal G_h(\bs \beta)=0}
\|
\mathsf A_h
\|,
\end{equation}
where
\begin{equation}\label{eq:aqec_A_CE}
\mathsf A_h
    :=
    (
        \mathbf h\, I+\mathfrak h\,\bs{L}
    )^\dagger
    (
        \mathbf h\, I+\mathfrak h\,\bs{L}
    ),
\end{equation}
and
\begin{equation}
\mathcal G_h(\bs\beta)
:=
G(\bs\beta)
+
h_0 I
+
\mathbf h^\dagger\bs{L}
+
\bs{L}^\dagger\mathbf h
+
\bs{L}^\dagger\mathfrak h\bs{L}.
\end{equation}
$\Gamma^{*-1}$ is the solution to the minimization problem and $\|\cdot\|$ denotes the spectral norm. $h_0\in \mathbb{R}$, $\mathbf h\in \mathbb{C}^r$, and the hermitian matrix $\mathfrak h\in\mathbb{C}^{r\times r}$ parameterize the freedom in representing the effective generator within the Lindblad span. If the feasible set in \cref{eq:Fq_opt_aqec} is empty, then no admissible embedding $G(\bs\beta)$ lies entirely in the Lindblad span. This is precisely the regime in which the functional HNLS condition holds and exact QEC can recover HL scaling. Thus, this bound deals with the complementary case in which the feasible set is nonempty and the optimal asymptotic scaling is SQL.

The construction that attains \cref{eq:Fq_opt_aqec} follows the single-parameter AQEC protocol of Ref.~\cite{ZhouJiang2020AQEC}. One first derives the dual of the SDP in \cref{eq:Fq_opt_aqec} and uses its optimal variables to construct a two-dimensional approximate code. We defer the details of this construction to  \cref{app:dual_SQL_AQEC} and state here only the resulting logical dynamics. Under sufficiently frequent recovery operations, the encoded state evolves according to
\begin{equation}
\frac{d\rho_L}{ds}
\simeq
-i
\epsilon\left[
q
Z_L,
\rho_L
\right]
+
\epsilon^2
\Gamma^*
\mathcal{D}_{Z_L}(\rho_L).
\label{eq:Lind_eff_q}
\end{equation}
Thus, the desired function again couples to a logical $Z_L$, but the logical qubit now experiences dephasing noise [\cref{fig:qcs}~(b)]. For the AQEC construction, the parameter $\epsilon\ll1$, meaning that the effective signal is suppressed to order $\epsilon$, while the dephasing rate is suppressed to order $\epsilon^2$. 

For the effective dephasing model in \cref{eq:Lind_eff_q}, an appropriate spin-squeezed logical probe along with swap operations asymptotically attains the optimal Fisher information rate $1/\Gamma^*$ and therefore saturates \cref{eq:Fq_opt_aqec} \cite{ZhouJiang2020AQEC}. We note that the parameter $\epsilon$ does not affect the asymptotic Fisher information rate because it weakens the signal and noise in a compensating manner. It does, however, determine the time required to reach the asymptotic regime which grows as $O(\epsilon^{-2})$. Consequently, the AQEC bound is approached only at increasingly long sensing times as $\epsilon\rightarrow0$.

\begin{table}[t]
    \centering
    \begin{tabular}{ccc}
        \hline\hline
        \, Noise \,
        &
        \, Functional HNLS \,
        &
        \, Optimal QFI\,
        \\
        \hline
        $\times$
        &
        $\checkmark$
        &
        $F_{q,\mathrm{opt}} =t^2\xi^{*2} \quad$(\cref{eq:F_opt_dual})
        \\
        $\checkmark$
        &
        $\checkmark$
        &
        $F_{q,\mathrm{opt}}=t^2\zeta^{*2} \quad $(\cref{eq:F_opt_dual_markovian})
        \\
        $\checkmark$
        &
        $\times$
        &
        $F_{q,\mathrm{opt}}=t/\Gamma^* \quad $ (\cref{eq:Fq_opt_aqec})
        \\
        \hline\hline
    \end{tabular}
    \caption{
    Summary of the three function estimation regimes considered in this work.
    The quantities $\xi^*$, $\zeta^*$, and $\Gamma^*$ denote the optimal
    variables in the corresponding optimization problems.
    }\label{tab:function_qfi_regimes}
\end{table}

\section{Function estimation advantage}
\label{sec:function_estimation_advantage}

The three function estimation regimes considered in this work are summarized in Table~\ref{tab:function_qfi_regimes}. We now use these results to determine when directly estimating a function of several parameters outperforms first estimating the individual parameters and then classically computing the function, thereby demonstrating a quantum computational sensing advantage \cite{Khan2025, Sarovar2023, Allen2025}.

Let $F_{\theta_i,\mathrm{opt}}$ denote the optimal QFI for estimating $\theta_i$, in the presence of all remaining nuisance parameters. If we assume estimating all the parameters optimally, and combining them to estimate $q$ then the corresponding variance is $\sum_i\alpha_i^2/F_{\theta_i,\mathrm{opt}}$. By comparison, direct optimal function estimation gives a variance of $1/F_{q,\mathrm{opt}}$. We therefore define the function estimation advantage ratio by
\begin{equation}
\eta^2
:=
\left(
\sum_i
\frac{\alpha_i^2}{
F_{\theta_i,\mathrm{opt}}
}
\right)
F_{q,\mathrm{opt}}.
\label{eq:fe_advantage}
\end{equation}
Thus, $\eta$ is the ratio between the standard deviations of the individual-estimation and direct-estimation strategies. 

When the individually optimal protocols for each $\theta_j$ can be implemented simultaneously, they constitute a valid indirect protocol for estimating $q$. However, in a general multiparameter problem, the scalar optimal $F_{\theta_i,\mathrm{opt}}$ need not be simultaneously attainable \cite{Ragy2016,GoreckiDemkowicz2022}. In this case, \cref{eq:fe_advantage} should be regarded as an optimistic benchmark for individual estimation. Thus, $\eta$ is a lower bound for function estimation advantage: $\eta>1$ certifies an advantage, whereas $\eta\leq1$ does not rule out an advantage over the best jointly realizable individual-estimation strategy.

The ratio $\eta$ is independent of the sensing time $t$ only when the optimal QFI for the function and all the individual parameters have the same asymptotic time dependence. If the function satisfies the functional HNLS condition while one or more relevant individual parameters do not, direct function estimation can additionally exhibit an advantage in its time scaling. We illustrate these different scenarios using three representative examples. Here, we state only the main results and defer the detailed calculations to \cref{app:function_advantage_examples}.

\paragraph{Average.}
Consider a network of $m$ qubit sensors with Hamiltonian $\sum_i\theta_i Z_i$, as in Ref.~\cite{Eldredge2018}, where each sensor is subject to local bit-flip noise generated by $X_i$. The target is the collective sum $q=\sum_i\theta_i$, or equivalently the average up to an overall normalization. The full HNLS condition holds in this example, and hence both the individual parameters and every linear function of them can retain HL scaling under QEC. Nevertheless, direct estimation of the collective function achieves an advantage $\eta=\sqrt{m}$ over the individual-estimation benchmark. Error correction therefore restores HL scaling in time while preserving the advantage in qubit number for estimating the average in the presence of local noise.

\paragraph{Difference.}
Next, consider two qubits subject to common-mode dephasing generated by $Z_1+Z_2$, with Hamiltonian $\theta_1Z_1+\theta_2Z_2$. The target function is the field difference $q=\theta_1-\theta_2$. In this case, the full HNLS condition fails because the common-mode parameter direction is contained in the Lindblad span. Consequently, neither individual parameter satisfies the functional HNLS condition. The difference $q$, however, is orthogonal to this invisible common-mode direction and therefore satisfies the functional HNLS condition. Direct estimation of $q$ retains HL scaling, whereas the individual-estimation benchmark is limited to SQL scaling, giving
$\eta=\sqrt{8\gamma t}$. This example demonstrates a time-scaling advantage in function estimation even in a system containing only two qubits.

\paragraph{Fourier mode.}
Finally, consider the Hamiltonian $\sum_j\theta_jX_j$ with neighbor-correlated decay generated by
$L_j=\sqrt{\gamma}(\sigma_j^--\sigma_{j+1}^-)$ with $\sigma^-_j = \ketbra{0_j}{1_j}$. The target is the cosine Fourier component
$q=\sum_{j=0}^{m-1}\cos(2\pi j/m)\theta_j$. Here, the functional HNLS condition fails, so neither the direct nor the individual-estimation strategy can attain HL scaling. Nevertheless, AQEC can optimize their SQL prefactors and yields a function estimation advantage
$\eta=m/(\sqrt{2}\pi)$ for $m\gg1$. Thus, a linear advantage in the number of sensors can arise even when both strategies remain SQL limited.

Table~\ref{tab:function_advantage_examples} summarizes the three examples. We emphasize that all noise strengths and jump coefficients remain bounded as the number of sensors increases which means their advantages do not arise from increasing the microscopic signal or noise strengths with $m$. These examples exhibit three distinct enhancements: a HL-prefactor advantage proportional to $\sqrt{m}$, a time-scaling advantage proportional to $\sqrt{t}$, and an SQL-prefactor advantage proportional to $m$. We note that, when the individual parameters are SQL limited, the largest possible time-scaling advantage is at most $\eta=O(\sqrt{t})$.

\begin{table*}[t]
\centering
{
\begin{tabular}{ccccc}
    \hline\hline
    Hamiltonian
    &
    Lindblad jumps
    &
    Target function
    &
    Functional HNLS
    &
    Advantage
    \\
    \hline
    $\displaystyle\sum_{i=1}^m\theta_iZ_i$
    &
    $\displaystyle L_i=\sqrt{\gamma_i}X_i$
    &
    $\displaystyle q=\sum_{i=1}^m\theta_i$
    &
    $\checkmark$
    &
    $\displaystyle \eta=\sqrt m$
    \\
    $\displaystyle \theta_1Z_1+\theta_2Z_2$
    &
    $\displaystyle L=\sqrt{\gamma}(Z_1+Z_2)$
    &
    $\displaystyle q=\theta_1-\theta_2$
    &
    $\checkmark$ for $q$, $\times$ for $\theta_i$
    &
    $\displaystyle \eta=\sqrt{8\gamma t}$
    \\
    $\displaystyle \sum_{j=0}^{m-1}\theta_jX_j$
    &
    $\displaystyle L_j=\sqrt{\gamma}
    \left(\sigma_j^--\sigma_{j+1}^-\right)$
    &
    $\displaystyle q=
    \sum_{j=0}^{m-1}
    \cos\left(\frac{2\pi j}{m}\right)\theta_j$
    &
    $\times$
    &
    $\displaystyle \eta=\frac{m}{\sqrt{2}\pi}$
    \\
    \hline\hline
\end{tabular}
}
\caption{
Examples of function estimation advantage. The examples respectively demonstrate an HL-prefactor advantage scaling as $\sqrt{m}$, a time-scaling advantage arising when only the target function satisfies functional HNLS, and an SQL-prefactor advantage scaling linearly with $m$ for $m\gg1$.
}
\label{tab:function_advantage_examples}

\end{table*}

\section{Outlook}

In this work, we derived tight bounds for estimating a linear function of Hamiltonian parameters in the presence of Markovian noise. We identified a necessary and sufficient functional HNLS condition for HL scaling, and constructed protocols that attain the optimal bounds. The examples in Sec.~\ref{sec:function_estimation_advantage} demonstrate quantum computational sensing advantages: directly estimating a target function can substantially outperform first learning all the Hamiltonian parameters and then classically evaluating the function.

We note that more general nonlinear target functions $q(\bs\theta)$ can be treated using the two-step adaptive protocols of Refs.~\cite{Abbasgholinejad2026,Qian2021} when the full multiparameter HNLS condition holds \cite{Gorecki2020OptimalProbes}. A coarse estimate can be obtained using a negligible fraction of the sensing time and then used to linearize the problem, mapping it to the setting solved here. Whether this approach extends to cases in which the individual parameters are SQL-limited remains open.

A central practical challenge is to remove the requirement for noiseless ancillas and fast control \cite{Rojkov2022Bias}. Ancilla-free codes have been constructed when the signal and noise commute \cite{Layden2019AncillaFree}, and more general multi-probe constructions can asymptotically attain optimal metrological coefficients with no ancillas or a negligible ancillary overhead \cite{Zhou2024AncillaFree}. These results, however, do not establish full fault tolerance. Recent work has begun to address this problem. Refs.~\cite{Sahu2026FaultTolerant, Conlon2026FaultTolerant} established fault-tolerant protocols assuming biased noise. Developing tight bounds under generic noise remains an important open question.

Another important extension is the simultaneous estimation of multiple functions. Deriving tight and attainable bounds for an arbitrary collection of functions would bridge the gap between single function estimation and full Hamiltonian learning. Such a theory could also clarify the ultimate asymptotic limits of Hamiltonian learning algorithms for estimating all the encoded parameters.

\begin{acknowledgments}
We thank Sisi Zhou for useful discussion. This project used ChatGPT for coming up with proof ideas and methods, as well as for general checking and proofreading. E.A., and A.V.G.~were supported in part by the ONR MURI, AFOSR MURI,  DoE ASCR Quantum Testbed Pathfinder program (award No.~DE-SC0024220), NSF QLCI (award No.~OMA-2120757), NSF STAQ program, ARL (W911NF-24-2-0107), and NQVL:QSTD:Design:FTL. E.A., and A.V.G.~also acknowledge support from the U.S.~Department of Energy, Office of Science, National Quantum Information Science Research Centers, Quantum Systems Accelerator (award No.~DE-SCL0000121) and from the U.S.~Department of Energy, Office of Science, Accelerated Research in Quantum Computing, Fundamental Algorithmic Research toward Quantum Utility (FAR-Qu). S.R.M. is supported by the NSF QLCI award OMA-2120757. Y.-X.W.~acknowledges support from a QuICS Hartree Postdoctoral Fellowship. J.B. notes that the views expressed in this work are those of the author and do not reflect the official policy of the United States Naval Academy or any department of the United States government.
\end{acknowledgments}

\bibliography{ref}

\newpage
\onecolumngrid
\appendix
\crefalias{section}{appendix}
\clearpage

\section{Dual formulation of the HL multiparameter bound}
\label{app:dual_HL_bound}

In this appendix, we derive the dual formulation of the noisy HL bound in
\cref{eq:F_q_HL}. The noiseless result follows as the special case
$\mathcal S_L=\mathrm{span}_{\mathbb R}\{I\}$.

The single-parameter bound of Ref.~\cite{Zhou2018} and \cref{eq:F_q} gives
\begin{equation}
    \lim_{t\rightarrow\infty}
    \frac{F_q}{t^2}
    \leq
    (\zeta^*)^2.
\label{eq:app_HL_bound_xi}
\end{equation}
with
\begin{equation}
    \zeta^*
    :=
    \min_{\bs\beta:\,\bs\alpha\cdot\bs\beta=1}
    \min_{S\in\mathcal S_L}
    \|G(\bs\beta)-S\|_s.
\label{eq:app_xi_primal_definition}
\end{equation}
We note that the bound in Ref.~\cite{Zhou2018} uses the operator norm instead of seminorm. Since $I\in\mathcal S_L$, the two forms are equivalent. Indeed, for any Hermitian $O$,
\begin{equation}
    2\min_{c\in\mathbb R}\|O-cI\|
    =
    \|O\|_s.
\end{equation}

To express \cref{eq:app_xi_primal_definition} as an SDP, introduce two
real variables $u$ and $v$. The identity
\begin{equation}
    \|A\|_s
    =
    \min_{u,v}
    \left\{
        u-v:
        vI\preceq A\preceq uI
    \right\}
\end{equation}
gives
\begin{align}
    \zeta^*
    =
    \min_{\bs\beta,S,u,v}\quad
    &u-v
\label{eq:app_HL_primal_SDP}\\
    \mathrm{s.t.}\quad
    &G(\bs\beta)-S\preceq uI,
    \nonumber\\
    &G(\bs\beta)-S\succeq vI,
    \nonumber\\
    &\bs\alpha\cdot\bs\beta=1,
    \qquad
    S\in\mathcal S_L.
    \nonumber
\end{align}

Let $\rho_0,\rho_1\succeq0$ be the dual variables associated with the
upper and lower operator inequalities, respectively, and let
$\zeta\in\mathbb R$ be the Lagrange multiplier associated with
$\bs\alpha\cdot\bs\beta=1$. Thus, the Lagrangian is
\begin{align}
    \mathcal L
    ={}&
    u-v
    +
    \mathrm{Tr}\!\left[
        \rho_0
        \left(
            G(\bs\beta)-S-uI
        \right)
    \right]
    +
    \mathrm{Tr}\!\left[
        \rho_1
        \left(
            vI-G(\bs\beta)+S
        \right)
    \right]
    +
    \zeta
    \left(
        1-\bs\alpha\cdot\bs\beta
    \right).
\label{eq:app_HL_Lagrangian}
\end{align}
The infimum over $u$ and $v$ is finite only if
\begin{equation}
    \mathrm{Tr}(\rho_0)
    =
    \mathrm{Tr}(\rho_1)
    =
    1.
\end{equation}
The infimum over $S\in\mathcal S_L$ is finite only if
\begin{equation}
    \mathrm{Tr}((\rho_0-\rho_1)S)=0,
    \qquad
    \forall S\in\mathcal S_L,
\label{eq:app_dual_lindblad_orthogonality}
\end{equation}
whereas the infimum over each component of $\bs\beta$ is finite only if
\begin{equation}
    \mathrm{Tr}((\rho_0-\rho_1)g_j)-\zeta\alpha_j=0,
    \qquad
    j=1,\ldots,m.
\end{equation}
Thus, the dual SDP is
\begin{align}
    \zeta^*
    =
    \max_{\zeta,\rho_0,\rho_1}\quad
    &\zeta
\label{eq:app_HL_dual_SDP}\\
    \mathrm{s.t.}\quad
    &\mathrm{Tr}\!\left[
        (\rho_0-\rho_1)g_j
    \right]
    =
    \zeta\alpha_j,
    \qquad
    j=1,\ldots,m,
    \nonumber\\
    &\mathrm{Tr}\!\left[
        (\rho_0-\rho_1)S
    \right]
    =
    0,
    \qquad
    \forall S\in\mathcal S_L,
    \nonumber\\
    &\rho_0,\rho_1\succeq0,
    \qquad
    \mathrm{Tr}(\rho_0)
    =
    \mathrm{Tr}(\rho_1)
    =
    1.
    \nonumber
\end{align}
Consequently,
\begin{equation}
    F_{q,\mathrm{opt}}
    =
    t^2\zeta^{*2}
\label{eq:app_HL_coefficient_dual}
\end{equation}
whenever functional HNLS holds. If functional HNLS fails, then
$\zeta^*=0$.

We next construct a code attaining
\cref{eq:app_HL_coefficient_dual}. Let $\rho_0^*$ and $\rho_1^*$ be
optimal solutions of \cref{eq:app_HL_dual_SDP}, and let
$\ket{\psi_0}$ and $\ket{\psi_1}$ be arbitrary purifications of these
states on the system and a noiseless ancilla. Introduce an additional
noiseless flag qubit $F$ and define
\begin{equation}
    \ket{0_L}
    :=
    \ket{\psi_0}\ket{0}_F,
    \qquad
    \ket{1_L}
    :=
    \ket{\psi_1}\ket{1}_F.
\label{eq:app_HL_codewords}
\end{equation}
The orthogonal flag states ensure that, for every system operator $O$,
\begin{equation}
    \bra{a_L}O\ket{b_L}
    =
    \delta_{ab}\mathrm{Tr}(\rho_a^*O).
\label{eq:app_flag_matrix_elements}
\end{equation}
Consequently, with
\begin{equation}
    \Pi_C
    =
    \ketbra{0_L}{0_L}
    +
    \ketbra{1_L}{1_L},
    \qquad
    Z_L
    =
    \ketbra{0_L}{0_L}
    -
    \ketbra{1_L}{1_L},
\end{equation}
we have
\begin{equation}
    \Pi_C O\Pi_C
    =
    \frac{\mathrm{Tr}[(\rho_0^*+\rho_1^*)O]}{2}\Pi_C
    +
    \frac{\mathrm{Tr}[(\rho_0^*-\rho_1^*)O]}{2}Z_L.
\label{eq:app_general_projected_operator}
\end{equation}

Applying \cref{eq:app_general_projected_operator} to $g_j$ and using the
first set of dual constraints gives
\begin{equation}
    \Pi_C g_j\Pi_C
    =
    \mu_j\Pi_C
    +
    \frac{\zeta^*\alpha_j}{2}Z_L,
\label{eq:app_projected_generators_dual}
\end{equation}
where
\begin{equation}
    \mu_j
    :=
    \frac{\mathrm{Tr}[(\rho_0^*+\rho_1^*)g_j]}{2}.
\end{equation}
More generally,
\begin{equation}
    \Pi_C G(\bs v)\Pi_C
    =
    \mu(\bs v)\Pi_C
    +
    \frac{\zeta^*}{2}
    (\bs\alpha\cdot\bs v)Z_L,
\label{eq:app_projected_generator_direction}
\end{equation}
where $\mu(\bs v) = \sum_j \mu_j v_j$. Every nuisance direction obeys $\bs\alpha\cdot\bs v=0$ and therefore
acts as a scalar within the code. The full signal Hamiltonian becomes
\begin{equation}
    \Pi_C H_0\Pi_C
    =
    \mu(\bs\theta)\Pi_C
    +
    \frac{\zeta^*q}{2}Z_L.
\label{eq:app_HL_logical_Hamiltonian}
\end{equation}

The second set of dual constraints implies
\begin{equation}
    \Pi_C S\Pi_C
    \propto
    \Pi_C,
    \qquad
    \forall S\in\mathcal S_L.
\end{equation}
Since $\mathcal S_L$ contains the Hermitian and anti-Hermitian parts of
$L_a$ and $L_a^\dagger L_b$, this is equivalent to
\begin{equation}
    \Pi_C L_a\Pi_C
    =
    \ell_a\Pi_C,
    \qquad
    \Pi_C L_a^\dagger L_b\Pi_C
    =
    c_{ab}\Pi_C,
\label{eq:app_HL_KL_conditions}
\end{equation}
for all $a,b$. Thus, the code satisfies the Knill--Laflamme conditions
for the infinitesimal Markovian error set $E_0=I$ and $E_a=L_a$.

For completeness, over a short interval $\delta t$, the channel has
Kraus operators
\begin{equation}
    K_0
    =
    I-
    \left(
        iH_0+\frac{1}{2}\sum_aL_a^\dagger L_a
    \right)\delta t
    +
    O(\delta t^2)
\end{equation}
and
\begin{equation}
    K_a
    =
    \sqrt{\delta t}\,L_a
    +
    O(\delta t^{3/2}).
\end{equation}
By \cref{eq:app_HL_KL_conditions}, there exists a recovery map
$\mathcal R$ such that, for every state supported on $C$,
\begin{equation}
    \mathcal R\circ\mathcal E_{\delta t}(\rho_L)
    =
    \rho_L
    -
    i\delta t
    \left[
        \Pi_C H_0\Pi_C,\rho_L
    \right]
    +
    O(\delta t^2).
\end{equation}
Applying the recovery after every interval and taking
$\delta t\rightarrow0$ suppresses the noise and Hamiltonian leakage while
preserving \cref{eq:app_HL_logical_Hamiltonian}. Finally, preparing $\ket{+_L}=(\ket{0_L}+\ket{1_L})/\sqrt{2}$ and evolving under the effective Hamiltonian for time $t$, gives
\begin{equation}
    F_q
    =
    4t^2
    \mathrm{var}_{\ket{+_L}}
    \left(
        \frac{\zeta^*}{2}Z_L
    \right)
    =
    t^2(\zeta^*)^2,
\end{equation}
which attains the HL coefficient in \cref{eq:app_HL_bound_xi}.

\section{Static-field protocol error analysis}
\label{app:energy_penalty}

In this appendix, we analyze the static protocol introduced in the main text. This
protocol applies only to the noiseless setting: it suppresses coherent
leakage generated by nuisance Hamiltonian terms, but it does not correct
stochastic Lindblad errors. We assume that the optimal coefficient in \cref{eq:F_opt_dual} satisfies $\xi^*>0$.

Let
\begin{equation}
    P:=\Pi_C,
    \qquad
    Q:=I-P,
\end{equation}
and suppose the unknown parameters are bounded such that
\begin{equation}
    \|H_0(\bs\theta)\|
    \leq
    \Lambda.
\label{eq:app_H0_uniform_bound}
\end{equation}

By construction of the noiseless code,
\begin{equation}
    PH_0P
    =
    \mu(\bs\theta)P
    +
    \frac{\xi^*q}{2}Z_L,
\label{eq:app_projected_H0_static}
\end{equation}
where $\mu(\bs\theta)$ contributes only a global phase within the code.

Add the static energy penalty
\begin{equation}
    H_c=BP,
\end{equation}
with $B>0$. The total Hamiltonian has the block form
\begin{equation}
    H_B
    :=
    H_0+BP
    =
    \begin{pmatrix}
        [B+\mu(\bs\theta)]I_C
        +\dfrac{\xi^*q}{2}Z_L
        &
        PH_0Q
        \\
        QH_0P
        &
        QH_0Q
    \end{pmatrix}.
\label{eq:app_HB_block_form}
\end{equation}
The corresponding block-diagonal Hamiltonian is
\begin{equation}
    H_Z
    :=
    PH_0P+QH_0Q.
\label{eq:app_Zeno_Hamiltonian}
\end{equation}

Then, the strong-coupling bound of
Refs.~\cite{Burgarth2019StrongCoupling,Burgarth2022OneBound}, applied to
the two eigenspaces of $P$, gives
\begin{equation}
    \left\|
        e^{-is(BP+H_0)}
        -
        e^{-is(BP+H_Z)}
    \right\|
    \leq
    \varepsilon_B(s),
\label{eq:app_static_strong_coupling_bound}
\end{equation}
where, for $0\leq s\leq t$,
\begin{equation}
    \varepsilon_B(s)
    :=
    \frac{2\sqrt{2}\,\Lambda}{B}
    \left(
        1+2s\Lambda
    \right).
\label{eq:app_static_epsilon}
\end{equation}
Thus the exact evolution converges in operator norm to the block-diagonal
evolution as $B\rightarrow\infty$.

For the logical input
$\ket{+_L}=(\ket{0_L}+\ket{1_L})/\sqrt{2}$, define
\begin{equation}
    \ket{\psi_B(t)}
    :=
    e^{-itH_B}\ket{+_L}
\end{equation}
and
\begin{equation}
    \ket{\psi_{\mathrm{eff}}(t)}
    :=
    e^{-it\xi^*qZ_L/2}\ket{+_L}.
\end{equation}
Because $\ket{+_L}$ is supported on $C$,
\begin{equation}
    e^{-it(BP+H_Z)}\ket{+_L}
    =
    e^{-it[B+\mu(\bs\theta)]}
    \ket{\psi_{\mathrm{eff}}(t)}.
\end{equation}
It follows from \cref{eq:app_static_strong_coupling_bound} that
\begin{equation}
    \min_{\varphi\in\mathbb R}
    \left\|
        \ket{\psi_B(t)}
        -
        e^{i\varphi}\ket{\psi_{\mathrm{eff}}(t)}
    \right\|
    \leq
    \varepsilon_B(t),
\label{eq:app_static_state_error}
\end{equation}
where $\varphi$ is a global phase. Finally, following the method of Ref.~\cite{Abbasgholinejad2025}, the change in the variance and bias of $q_\text{est}$ remains negligible for $B = \Omega(t\|H_0\|^2)$ after measuring the parity operator $i\ketbra{0_L}{1_L}+\text{h.c.}$.

\section{Dual formulation of the SQL AQEC bound}
\label{app:dual_SQL_AQEC}

In this appendix, we derive the code-construction dual of the SQL bound in
\cref{eq:Fq_opt_aqec}. Throughout this appendix, we assume that the
functional HNLS condition fails and that the feasible set in
\cref{eq:Fq_opt_aqec} is nonempty. As a reminder
\begin{equation}
    \frac{1}{\Gamma^*}
    :=
    4
    \min_{\bs\beta:\,\bs\alpha\cdot\bs\beta=1}
    \min_{h_0,\mathbf h,\mathfrak h:\,
        \mathcal G_h(\bs\beta)=0}
    \|\mathsf A_h\|,
\label{eq:app_AQEC_rate_primal}
\end{equation}
where
\begin{equation}
\mathsf A_h
    :=
    (
        \mathbf h\, I+\mathfrak h\,\bs{L}
    )^\dagger
    (
        \mathbf h\, I+\mathfrak h\,\bs{L}
    ),
\end{equation}
and
\begin{equation}
\mathcal G_h(\bs\beta)
:=
G(\bs\beta)
+
h_0 I
+
\mathbf h^\dagger\bs{L}
+
\bs{L}^\dagger\mathbf h
+
\bs{L}^\dagger\mathfrak h\bs{L}.
\end{equation}
and, the bound in the main text states
\begin{equation}
    \lim_{t\rightarrow\infty}
    \frac{F_{q,\mathrm{opt}}}{t}
    \leq
    \frac{1}{\Gamma^*}.
\end{equation}
Here, $h_0\in\mathbb R$, $\mathbf h\in\mathbb C^r$, and
$\mathfrak h\in\mathbb C^{r\times r}$ is Hermitian. 

For every positive operator $A$,
\begin{equation}
    \|A\|
    =
    \max_{\rho\succeq0,\,\mathrm{Tr}(\rho)=1}
    \mathrm{Tr}(\rho A).
\end{equation}
Write $\rho=CC^\dagger$, where $C$ is a $d\times d$ matrix satisfying $\mathrm{Tr}(C^\dagger C)=1$ and $d$ is the system dimension. Using the same minimax step as in Ref.~\cite{ZhouJiang2020AQEC}, now including the extra constraint on $\bs\beta$, gives
\begin{equation}
    \frac{1}{\Gamma^*}
    =
    \max_{C:\,\mathrm{Tr}(C^\dagger C)=1}
    \;
    \min_{\substack{\bs\beta,h_0,\mathbf h,\mathfrak h\\
        \mathcal G_h(\bs\beta)=0,\;
        \bs\alpha\cdot\bs\beta=1}}
    4\mathrm{Tr}
    \left(
        C^\dagger\mathsf A_h C
    \right).
\label{eq:app_AQEC_minimax}
\end{equation}
For fixed $C$, we first perform the same gauge transformation of the jumps $\{L_a\}_{a=1}^r$ as in Ref.~\cite{ZhouJiang2020AQEC}. We subtract their $C$-expectation values and then apply a unitary to obtain operators
$\{J_a\}_{a=1}^r$ satisfying
\begin{equation}
    \mathrm{Tr}(C^\dagger J_aC)=0,
    \qquad
    \mathrm{Tr}
    \left(
        C^\dagger J_a^\dagger J_bC
    \right)
    =
    \kappa_a\delta_{ab},
\label{eq:app_AQEC_jump_gauge}
\end{equation}
where $\kappa_a\geq0$, and we define $\mathfrak n_C:=\{a:\,\kappa_a=0\}$. This transformation only induces a parameter-independent Hamiltonian shift that can be
cancelled by the control term $H_c(t)$. In the remainder of the derivation, $\mathbf L$ is replaced by $\mathbf J$ in $\mathsf A_h$ and $\mathcal{G}_h(\bs\beta)$.

In this gauge, the quadratic objective simplifies to
\begin{equation}
    \mathrm{Tr}
    \left(
        C^\dagger\mathsf A_hC
    \right)
    =
    \mathbf h^\dagger\mathbf h
    +
    \mathrm{Tr}
    \left(
        \mathrm{diag}(\bs\kappa)\,
        \mathfrak h^2
    \right),
\label{eq:app_AQEC_quadratic_objective}
\end{equation}
where $\bs\kappa=(\kappa_1,\ldots,\kappa_r)$. 

Introduce a Hermitian matrix $\widetilde C$ as the Lagrange multiplier
for the operator equality $\mathcal G_h(\bs\beta)=0$, and a real scalar
$\lambda$ for $\bs\alpha\cdot\bs\beta=1$. Thus, the Lagrangian of the inner problem in
\cref{eq:app_AQEC_minimax} is
\begin{align}
    \mathcal L
    ={}&
    4\mathrm{Tr}
    \left(
        C^\dagger\mathsf A_h C
    \right)
    +
    \mathrm{Tr}
    \left[
        \widetilde C\,
        \mathcal G_h(\bs\beta)
    \right]+
    \lambda
    \left(
        1-\bs\alpha\cdot\bs\beta
    \right).
\label{eq:app_AQEC_Lagrangian}
\end{align}
The infimum over $h_0$ requires
$\mathrm{Tr}(\widetilde C)=0$. The infimum over
$\bs\beta$ requires 
\begin{equation}
    \mathrm{Tr}(g_j\widetilde C)
    =
    \lambda\alpha_j,
    \qquad
    j=1,\ldots,m,
\label{eq:app_AQEC_response_alignment}
\end{equation}
If $\kappa_a=\kappa_b=0$, the infimum over the
corresponding unconstrained components of $\mathfrak h$ is finite only
if
\begin{equation}
    \mathrm{Tr}
    \left(
        J_a^\dagger J_b\widetilde C
    \right)
    =0,
    \qquad
    a,b\in\mathfrak n_C.
\label{eq:app_AQEC_null_constraints}
\end{equation}

Minimizing the remaining quadratic expression over $\mathbf h$ and
$\mathfrak h$ gives the residual-noise functional
\begin{align}
    \Gamma_C(\widetilde C)
    :=&
    \sum_{a=1}^r
    \left|
        \mathrm{Tr}(J_a\widetilde C)
    \right|^2+
    \sum_{a,b:\,\kappa_a+\kappa_b\neq0}
    \frac{
        \left|
            \mathrm{Tr}
            \left(
                J_a^\dagger J_b\widetilde C
            \right)
        \right|^2
    }{
        2(\kappa_a+\kappa_b)
    }.
\label{eq:app_AQEC_Gamma}
\end{align}
More explicitly, after the minimization the dual objective is
\begin{equation}
    \lambda-
    \frac{1}{4}
    \Gamma_C(\widetilde C).
\label{eq:app_AQEC_unscaled_dual_objective}
\end{equation}
Rescaling $\widetilde C$ rescales $\lambda$ linearly and
$\Gamma_C$ quadratically. Maximizing over this scale transforms
\cref{eq:app_AQEC_unscaled_dual_objective} into 
\begin{align}
    \frac{1}{\Gamma^*}
    =
    \max_{C,\widetilde C,\lambda}
    \quad&
    \frac{\lambda^2}{\Gamma_C(\widetilde C)}
\label{eq:app_AQEC_dual_quotient}\\
    \mathrm{s.t.}\quad
    &\mathrm{Tr}(C^\dagger C)=1,
    \nonumber\\
    &\widetilde C=\widetilde C^\dagger,
    \qquad
    \mathrm{Tr}(\widetilde C)=0,
    \nonumber\\
    &\mathrm{Tr}(g_j\widetilde C)
    =\lambda\alpha_j,
    \qquad
    j=1,\ldots,m,
    \nonumber\\
    &\mathrm{Tr}
    \left(
        J_a^\dagger J_b\widetilde C
    \right)
    =0,
    \qquad
    a,b\in\mathfrak n_C.
    \nonumber
\end{align}
which is invariant under a nonzero rescaling of
$\widetilde C$. By absorbing the coefficient $\lambda$ into $\tilde{C}$ for $\lambda\neq 0$, we obtain
\begin{align}
    \Gamma^*
    :=
    \min_{C,\widetilde C}
    \quad&
    \Gamma_C(\widetilde C)
\label{eq:app_AQEC_Gamma_star}\\
    \mathrm{s.t.}\quad
    &\mathrm{Tr}(C^\dagger C)=1,
    \qquad
    \widetilde C=\widetilde C^\dagger,
    \qquad
    \mathrm{Tr}(\widetilde C)=0,
    \nonumber\\
    &\mathrm{Tr}(g_j\widetilde C)=\alpha_j,
    \qquad
    j=1,\ldots,m,
    \nonumber\\
    &\mathrm{Tr}
    \left(
        J_a^\dagger J_b\widetilde C
    \right)
    =0,
    \qquad
    a,b\in\mathfrak n_C.
    \nonumber
\end{align}
We note that all these steps are the same as in Ref.~\cite{ZhouJiang2020AQEC} except for \cref{eq:app_AQEC_response_alignment} which arises from the extra freedom on the vector $\bs\beta$ for our multiparameter problem.

We now review the construction of the AQEC code in Ref.~\cite{ZhouJiang2020AQEC} using the optimal variables $(C^*,\widetilde C^*)$ solving the optimization problem in
\cref{eq:app_AQEC_Gamma_star}. Choose a matrix $D^*$ such that
\begin{equation}
    \widetilde C^*
    =
    C^*D^{*\dagger}+D^*C^{*\dagger},
    \qquad
    \mathrm{Tr}(C^{*\dagger}D^*)=0.
\label{eq:app_AQEC_D_definition}
\end{equation}
If $C^*$ is invertible, one may take
$D^{*\dagger}=C^{*-1}\widetilde C^*/2$. A singular $C^*$ may be
approximated arbitrarily closely by a full-rank matrix as detailed in 
\cite{ZhouJiang2020AQEC}.

Let $A'$ be a $d$-dimensional noiseless ancilla and let $F$ be a
noiseless flag qubit. For $\epsilon \ll 1$, define
\begin{equation}
    \mathcal N_\epsilon
    :=
    \left[
        1+\epsilon^2
        \mathrm{Tr}(D^{*\dagger}D^*)
    \right]^{-1/2}
\end{equation}
and
\begin{subequations}\label{eq:app_AQEC_codewords}
\begin{align}
    \ket{0_L}
    &:=
    \mathcal N_\epsilon
    \sum_{i,j}
    (C^*_{ij}+\epsilon D^*_{ij})
    \ket{i}_S\ket{j}_{A'}\ket{0}_F,
\label{eq:app_AQEC_codeword_zero}\\
    \ket{1_L}
    &:=
    \mathcal N_\epsilon
    \sum_{i,j}
    (C^*_{ij}-\epsilon D^*_{ij})
    \ket{i}_S\ket{j}_{A'}\ket{1}_F.
\label{eq:app_AQEC_codeword_one}
\end{align}
\end{subequations}
The orthogonal flag makes all system operators diagonal in the logical
basis, while the opposite perturbations $\pm\epsilon D^*$ create a
logical signal of order $\epsilon$.

It remains to specify the recovery. Let $P$ project onto the code in
\cref{eq:app_AQEC_codewords} and let $Q=I-P$. Since every $J_a$
acts trivially on the flag qubit, define vectors on $S\otimes A'$ by
\begin{equation}
    Q J_a\ket{\ell_L}
    =
    \ket{e_{a,\ell}}\ket{\ell}_F,
    \qquad
    \ell=0,1.
\end{equation}
Form the error-overlap operator
\begin{equation}
    M_\epsilon
    :=
    \sum_{a=1}^r
    \ket{e_{a,0}}\bra{e_{a,1}}
\label{eq:app_AQEC_recovery_overlap}
\end{equation}
and choose a singular-value decomposition
\begin{equation}
    M_\epsilon
    =
    \sum_\mu s_\mu
    \ket{R_\mu}\bra{S_\mu}.
\end{equation}
The recovery channel has Kraus operators
\begin{equation}
    K_\mu
    :=
    \ket{0_L}\bra{R_\mu,0}
    +
    \ket{1_L}\bra{S_\mu,1},
\label{eq:app_AQEC_recovery_Kraus}
\end{equation}
so that
\begin{equation}
    \mathcal R_\epsilon(X)
    =
    \sum_\mu K_\mu X K_\mu^\dagger.
\label{eq:app_AQEC_recovery_map}
\end{equation}
The complete AQEC step is the CPTP map
\begin{equation}
    \mathcal V_\epsilon(X)
    :=
    PXP
    +
    \mathcal R_\epsilon
    \left(
        Q XQ
    \right).
\label{eq:app_AQEC_complete_step}
\end{equation}

Now, using the constraint in \cref{eq:app_AQEC_Gamma_star}, and \cref{eq:app_AQEC_codewords}, we obtain
\begin{equation}
    \mathrm{Tr}(g_jZ_L)
    =
    2\epsilon\alpha_j
    +
    O(\epsilon^3),
\end{equation}
where $g_j$ acts on the sensing system and as the identity on the
ancillas. Therefore,
\begin{equation}
    \frac{1}{2}
    \mathrm{Tr}(H_0Z_L)
    =
    \epsilon q
    +
    O(\epsilon^3),
\label{eq:app_AQEC_logical_response}
\end{equation}
while every nuisance direction has vanishing logical response to this
order.

Ref.~\cite{ZhouJiang2020AQEC} proves that applying
$\mathcal V_\epsilon$ sufficiently frequently and cancelling the known,
parameter-independent logical Hamiltonian shift gives, to leading
nonvanishing order in $\epsilon$,
\begin{equation}
    \frac{d\rho_L}{ds}
    =
    -i\epsilon
    \left[
        qZ_L,\rho_L
    \right]
    +
    \epsilon^2\Gamma^*
    \left(
        Z_L\rho_LZ_L-\rho_L
    \right)
    +
    o(\epsilon^2).
\label{eq:app_AQEC_effective_dynamics}
\end{equation}
For the effective dephasing channel in
\cref{eq:app_AQEC_effective_dynamics}, the asymptotically optimal
spin-squeezed logical probe attains Fisher information rate $1/\Gamma^*$ for estimating $q$ \cite{ZhouJiang2020AQEC}.

\section{Function estimation advantage examples}
\label{app:function_advantage_examples}

In this appendix, we derive the three examples summarized in
Sec.~\ref{sec:function_estimation_advantage}.
Recall that
\begin{equation}
    \eta^2
    =
    \left(
        \sum_i
        \frac{\alpha_i^2}{F_{\theta_i,\mathrm{opt}}}
    \right)
    F_{q,\mathrm{opt}}.
\label{eq:app_advantage_definition}
\end{equation}
The optimal QFI $F_{\theta_i,\mathrm{opt}}$ treats every parameter
other than $\theta_i$ as an unknown nuisance parameter. As explained in
the main text, the optimal QFIs need not be jointly attainable, so
\cref{eq:app_advantage_definition} is a lower bound for function estimation advantage.

\subsection{Average under local bit-flip noise}
\label{app:advantage_average}

Consider
\begin{equation}
    H_0
    =
    \sum_{i=1}^m\theta_iZ_i,
    \qquad
    L_i
    =
    \sqrt{\gamma_i}X_i.
\label{eq:app_average_model}
\end{equation}
The Lindblad span is generated by $I$, $X_i$, and $X_iX_j$. The
operators $Z_i$ are orthogonal to this span and linearly independent,
so the full HNLS condition holds.

For $G(\bs\beta)=\sum_i\beta_iZ_i$, we claim that
\begin{equation}
    \min_{S\in\mathcal S_L}
    \|G(\bs\beta)-S\|_s
    =
    2\|\bs\beta\|_1.
\label{eq:app_average_spectral_range}
\end{equation}
The upper bound follows by taking $S=0$. For the reverse inequality,
choose a computational-basis string $\bs z$ such that
$\bra{\bs z}Z_i\ket{\bs z}=\mathrm{sgn}(\beta_i)$ whenever
$\beta_i\neq0$, and let $\bar{\bs z}$ denote its bitwise complement.
Every nonidentity Pauli element of $\mathcal S_L$ has zero diagonal
expectation in both states, while the identity part has the same
expectation in both. Hence,
\begin{equation}
    \bra{\bs z}
    G(\bs\beta)-S
    \ket{\bs z}
    -
    \bra{\bar{\bs z}}
    G(\bs\beta)-S
    \ket{\bar{\bs z}}
   = 2 \|\bs \beta \|_1,
\end{equation}
which lower-bounds the spectral range.

H\"older's inequality gives
\begin{equation}
    1
    =
    |\bs\alpha\cdot\bs\beta|
    \leq
    \|\bs\alpha\|_\infty
    \|\bs\beta\|_1,
\end{equation}
and equality is obtained by placing all the weight of $\bs\beta$ on an
index maximizing $|\alpha_i|$. Hence
\begin{equation}
    \lim_{t\rightarrow\infty}
    \frac{F_{q,\mathrm{opt}}}{t^2}
    =
    \frac{4}{\|\bs\alpha\|_\infty^2},
    \qquad
    \lim_{t\rightarrow\infty}
    \frac{F_{\theta_i,\mathrm{opt}}}{t^2}
    =
    4.
\label{eq:app_average_QFIs}
\end{equation}
These coefficients can be attained with system-ancilla codes
$\ket{0_{L,i}}=\ket{0}_{S_i}\ket{0}_{A_i}$ and
$\ket{1_{L,i}}=\ket{1}_{S_i}\ket{1}_{A_i}$, which correct $X_i$ while
preserving $Z_i$ as a logical $Z$ operator.

Substitution into \cref{eq:app_advantage_definition} gives the general
result
\begin{equation}
    \eta^2
    =
    \frac{\|\bs\alpha\|_2^2}{\|\bs\alpha\|_\infty^2}.
\label{eq:app_average_eta_general}
\end{equation}
For the uniform sum $q=\sum_i\theta_i$, or for the average obtained by
dividing this sum by $m$, the overall normalization cancels from
$\eta$, and
\begin{equation}
    \eta=\sqrt m.
\label{eq:app_average_eta}
\end{equation}

\subsection{Difference under collective dephasing}
\label{app:advantage_gradient}

Consider two qubits with
\begin{equation}
    H_0
    =
    \theta_1Z_1+\theta_2Z_2,
    \qquad
    L
    =
    \sqrt{\gamma}(Z_1+Z_2).
\label{eq:app_gradient_model}
\end{equation}
Here,
\begin{equation}
    \mathcal S_L
    =
    \mathrm{span}_{\mathbb R}
    \{I,Z_1+Z_2,Z_1Z_2\},
\end{equation}
and therefore
\begin{equation}
    Z_{1,\perp}
    =
    \frac{Z_1-Z_2}{2},
    \qquad
    Z_{2,\perp}
    =
    -\frac{Z_1-Z_2}{2}.
\end{equation}
It follows that
\begin{equation}
    \ker(G_\perp)
    =
    \mathrm{span}\{(1,1)\}.
\end{equation}
The full multiparameter HNLS condition fails, but the target
$q=\theta_1-\theta_2$, with $\bs\alpha=(1,-1)$, satisfies functional
HNLS because it is orthogonal to the common-mode direction.

Every admissible embedding can be written as
\begin{equation}
    \bs\beta
    =
    \left(
        c+\frac{1}{2},
        c-\frac{1}{2}
    \right).
\end{equation}
Its common-mode component can be removed using $\mathcal S_L$, leaving
$(Z_1-Z_2)/2$, whose spectral range is two.  Moreover, every $S\in\mathcal S_L$ has the same
expectation value in $\ket{01}$ and $\ket{10}$, while
\begin{equation}
    \bra{01}\frac{Z_1-Z_2}{2}\ket{01}
    -
    \bra{10}\frac{Z_1-Z_2}{2}\ket{10}
    =
    2.
\end{equation}
Therefore,
\begin{equation}
    \|
        \frac{Z_1-Z_2}{2}-S
    \|_s
    \geq 2
\end{equation}
for all $S\in\mathcal S_L$. Thus
\begin{equation}
    \lim_{t\rightarrow\infty}
    \frac{F_{q,\mathrm{opt}}}{t^2}
    =
    4.
\label{eq:app_gradient_function_QFI}
\end{equation}
The result is attained by the code
\begin{equation}
    \ket{0_L}=\ket{01},
    \qquad
    \ket{1_L}=\ket{10},
\end{equation}
on which the signal Hamiltonian is $qZ_L$ and the jump vanishes.

By contrast, consider estimation of $\theta_1$ with $\theta_2$ unknown.
The condition $\beta_1=1$, together with
$\mathcal G_h(\bs\beta)=0$, forces $\beta_2=1$ and hence
$G(\bs\beta)=Z_1+Z_2$. For the single jump in
\cref{eq:app_gradient_model}, the minimizing variables are
\begin{equation}
    h
    =
    -\frac{1}{2\sqrt\gamma},
    \qquad
    h_0=0,
    \qquad
    \mathfrak h=0.
\end{equation}
They give
\begin{equation}
    \mathcal G_h=0,
    \qquad
    \mathsf A_h
    =
    \frac{I}{4\gamma}.
\end{equation}
The Pauli components in the constraint also show that these values are
necessary at the optimum. The same calculation applies to
$\theta_2$, so
\begin{equation}
    \lim_{t\rightarrow\infty}
    \frac{F_{\theta_1,\mathrm{opt}}}{t}
    =
    \lim_{t\rightarrow\infty}
    \frac{F_{\theta_2,\mathrm{opt}}}{t}
    =
    \frac{1}{\gamma}.
\label{eq:app_gradient_individual_QFIs}
\end{equation}
Combining \cref{eq:app_gradient_function_QFI} and
\cref{eq:app_gradient_individual_QFIs} yields
\begin{equation}
    \eta^2
    =
    8\gamma t.
\label{eq:app_gradient_eta}
\end{equation}
This is a scaling advantage in which the differential field is HL, whereas
either local field is only SQL when the other is an unknown nuisance.

\subsection{A general SQL formula for correlated amplitude damping}
\label{app:advantage_SQL_general}

Let
\begin{equation}
    H_0
    =
    \sum_{i=1}^m\theta_iX_i,
    \qquad
    L_a
    =
    \sum_{i=1}^m\ell_{ai}\sigma_i^-,
\label{eq:app_correlated_damping_model}
\end{equation}
where the coefficients $\ell_{ai}$ are real, and define
\begin{equation}
    K
    :=
    \ell^T\ell.
\label{eq:app_K_definition}
\end{equation}
We write $K^+$ for the Moore--Penrose pseudoinverse.

Consider an embedding
$G(\bs\beta)=\sum_i\beta_iX_i$. The terms $h_0I$ and
$\mathbf L^\dagger\mathfrak h\mathbf L$ preserve excitation number,
whereas $G(\bs\beta)$ and
$\mathbf h^\dagger\mathbf L+\mathbf L^\dagger\mathbf h$ change it by
one. The excitation-changing part of
$\mathcal G_h(\bs\beta)=0$ therefore requires
\begin{equation}
    \ell^T\mathrm{Re}(\mathbf h)
    =
    -\bs\beta,
    \qquad
    \ell^T\mathrm{Im}(\mathbf h)
    =
    0.
\label{eq:app_damping_feasibility_equations}
\end{equation}
This condition is feasible if and only if
$\bs\beta\in\mathrm{Range}(\ell^T)
=\mathrm{Range}(K)$.

Let $\ket{\Omega}=\ket{0}^{\otimes m}$ be the vacuum. Since
$L_a\ket{\Omega}=0$,
\begin{equation}
    \|\mathsf A_h\|
    \geq
    \bra{\Omega}\mathsf A_h\ket{\Omega}
    =
    \|\mathbf h\|_2^2
    \geq
    \bs\beta^TK^+\bs\beta.
\label{eq:app_damping_A_lower_bound}
\end{equation}
The last step is the minimum-norm solution of the first equation in
\cref{eq:app_damping_feasibility_equations}. When
$\bs\beta\in\mathrm{Range}(K)$, equality is attained by
\begin{equation}
    \mathbf h^*
    =
    -\ell K^+\bs\beta,
    \qquad
    h_0^*=0,
    \qquad
    \mathfrak h^*=0.
\end{equation}
Indeed, $\ell^T\mathbf h^*=-KK^+\bs\beta=-\bs\beta$ and
$\mathsf A_{h^*}=(\bs\beta^TK^+\bs\beta)I$. If
$F_{\bs\beta,\mathrm{opt}}$ denotes the optimal QFI for a scalar
parameter coupled through $G(\bs\beta)$, AQEC attainability then gives
\begin{equation}
    \lim_{t\rightarrow\infty}
    \frac{F_{\bs\beta,\mathrm{opt}}}{t}
    =
    4\bs\beta^TK^+\bs\beta,
    \qquad
    \bs\beta\in\mathrm{Range}(K).
\label{eq:app_damping_direction_rate}
\end{equation}
If $\bs\beta\notin\mathrm{Range}(K)$, the SQL constraint is
infeasible and \cref{eq:app_damping_direction_rate} does not apply.

Within the span of the $X_i$ generators, the directions contained in
the Lindblad span are exactly $\mathrm{Range}(K)$. The functional
HNLS condition is consequently
\begin{equation}
    \mathrm{Range}(K)
    \subseteq
    \ker(\bs\alpha^T)
    \quad\Longleftrightarrow\quad
    K\bs\alpha=0.
\label{eq:app_damping_functional_HNLS}
\end{equation}
Suppose instead that $K\bs\alpha\neq0$, so the function is SQL-limited.
For every admissible $\bs\beta\in\mathrm{Range}(K)$,
\begin{align}
    1
    &=(\bs\alpha^T\bs\beta)^2
    \nonumber\\
    &=
    \left[
        (K^{1/2}\bs\alpha)^T
        ((K^+)^{1/2}\bs\beta)
    \right]^2
    \nonumber\\
    &\leq
    (\bs\alpha^TK\bs\alpha)
    (\bs\beta^TK^+\bs\beta).
\label{eq:app_damping_CS_bound}
\end{align}
Equality is attained by
\begin{equation}
    \bs\beta^*
    =
    \frac{K\bs\alpha}{\bs\alpha^TK\bs\alpha}.
\end{equation}
Therefore,
\begin{equation}
    \lim_{t\rightarrow\infty}
    \frac{F_{q,\mathrm{opt}}}{t}
    =
    \frac{4}{\bs\alpha^TK\bs\alpha}.
\label{eq:app_damping_function_rate}
\end{equation}

For every relevant index with $K_{ii}>0$, the same formula applied to
$\bs\alpha=\bs e_i$ gives
\begin{equation}
    \lim_{t\rightarrow\infty}
    \frac{F_{\theta_i,\mathrm{opt}}}{t}
    =
    \frac{4}{K_{ii}}.
\label{eq:app_damping_individual_rate}
\end{equation}
If $K_{ii}=0$, positivity of $K$ implies $K\bs e_i=0$, so
$\theta_i$ instead satisfies functional HNLS and the SQL expression
does not apply. Assuming $K_{ii}>0$ for all indices with
$\alpha_i\neq0$, \cref{eq:app_advantage_definition} becomes
\begin{equation}
    \eta^2
    =
    \frac{
        \sum_i\alpha_i^2K_{ii}
    }{
        \bs\alpha^TK\bs\alpha
    }.
\label{eq:app_damping_eta_general}
\end{equation}
Thus, the SQL advantage compares the local loss strengths on the
diagonal of $K$ with the damping of the collective target profile.

\subsubsection{Fourier mode on a periodic chain}
\label{app:advantage_Fourier}

Take $m\geq3$, use indices modulo $m$, and let
\begin{equation}
    L_j
    =
    \sqrt\gamma
    \left(
        \sigma_j^- - \sigma_{j+1}^-
    \right),
    \qquad
    j=0,\ldots,m-1.
\label{eq:app_Fourier_jumps}
\end{equation}
Then
\begin{equation}
    K
    =
    \gamma(2I-T-T^\dagger),
\end{equation}
where $T$ translates by one site. For
\begin{equation}
    \alpha_j
    =
    \cos\left(\frac{2\pi j}{m}\right),
\end{equation}
we have
\begin{equation}
    K\bs\alpha
    =
    4\gamma\sin^2\left(\frac{\pi}{m}\right)\bs\alpha,
    \qquad
    \|\bs\alpha\|_2^2
    =
    \frac{m}{2}.
\end{equation}
Functional HNLS therefore fails for this nonzero Fourier mode. Using
\cref{eq:app_damping_function_rate} gives
\begin{equation}
    \lim_{t\rightarrow\infty}
    \frac{F_{q,\mathrm{opt}}}{t}
    =
    \frac{2}{
        \gamma m\sin^2(\pi/m)
    }.
\label{eq:app_Fourier_function_rate}
\end{equation}
Because $K_{jj}=2\gamma$,
\begin{equation}
    \lim_{t\rightarrow\infty}
    \frac{F_{\theta_j,\mathrm{opt}}}{t}
    =
    \frac{2}{\gamma}.
\end{equation}
The advantage at fixed $m$ is
\begin{equation}
    \eta
    =
    \frac{1}{\sqrt{2}\sin(\pi/m)},
\label{eq:app_Fourier_eta_exact}
\end{equation}
and hence
\begin{equation}
    \eta
    \sim
    \frac{m}{\sqrt{2}\pi},
    \qquad
    m\rightarrow\infty.
\label{eq:app_Fourier_eta_asymptotic}
\end{equation}
Thus, both strategies remain SQL-limited in time, but direct estimation
of the long-wavelength Fourier mode has a linear asymptotic advantage
in system size.

 \end{document}